\documentclass[twocolumn,journal]{IEEEtran}
\usepackage{verbatim}
\usepackage{amsmath}
\usepackage{amsthm}
\usepackage{amssymb}
\usepackage{color}
\usepackage{bm}
\usepackage{multirow}

\usepackage[pdftex]{graphicx}

\newtheorem{proposition}{Proposition} 
\newtheorem{theorem}{Theorem}
\newtheorem{corollary}{Corollary}
\newtheorem{remark}{Remark}
\newtheorem{lemma}{Lemma}

\def\Ker{\mathop{\rm Ker}}
\def\QED{\mbox{\rule[0pt]{1.5ex}{1.5ex}}}

\renewcommand{\qed}{\hfill \QED}

 \newenvironment{proofof}[1]{\vspace*{5mm} \par \noindent
         \quad{\it Proof of #1:\hspace{2mm}}}{\qed
}

\def\FF{\mathbb{F}}

\def\rank{\mathop{\rm rank}}
\def\im{\mathop{\rm Im}}
\def\rank{\mathop{\rm rank}}

\def\Label#1{\label{#1}\ [\ \text{#1}\ ]\ }
\def\Label{\label}

\begin{document}

\title{Asymptotically Secure Network Code for Active Attacks 
and its Application to Network Quantum Key Distribution
}

\author{Masahito Hayashi \IEEEmembership{Fellow, IEEE}
and Ning Cai \IEEEmembership{Fellow, IEEE}
\thanks{The work of M. Hayashi was supported in part by
the Japan Society of the Promotion of Science (JSPS) Grant-in-Aid 
for Scientific Research (B) Grant 16KT0017 and
for Scientific Research (A) Grant 17H01280 and
for Scientific Research (C) Grant 16K00014,
in part by the Okawa Research Grant, and 
in part by the Kayamori Foundation of Informational Science Advancement.
The material in this paper was presented in part at the 2017 IEEE International Symposium on Information Theory (ISIT 2017),   Aachen (Germany), 25-30 June 2017 \cite{HOKC}.}
\thanks{Masahito Hayashi is with the Graduate School of Mathematics, Nagoya University, Nagoya, 464-8602, Japan. 
He is also with 
Shenzhen Institute for Quantum Science and Engineering, Southern University of Science and Technology,
Shenzhen, 518055, China,
Center for Quantum Computing, Peng Cheng Laboratory, Shenzhen 518000, China,
and the Centre for Quantum Technologies, National University of Singapore, 3 Science Drive 2, 117542, Singapore
(e-mail:masahito@math.nagoya-u.ac.jp).
Ning Cai is with the School of Information Science and Technology, ShanghaiTech University, Middle Huaxia Road no 393,
Pudong, Shanghai  201210, China
(e-mail: ningcai@shanghaitech.edu.cn).} }

\markboth{M. Hayashi 
and N. Cai: 
Asymptotically Secure Network Code for Active Attacks 
}{}

\maketitle

\begin{abstract}
When there exists a malicious attacker in the network, we need to be careful of eavesdropping and
contamination. This problem is crucial for network communication when the network is realized by
a partially trusted relay of quantum key distribution. We discuss the asymptotic rate in a linear network
with the secrecy and robustness conditions when the above type of attacker exists. Also, under the
same setting, we discuss the asymptotic rate in a linear network when we impose the secrecy condition
alone. Then, we apply these results to the network composed of a partially trusted relay of quantum key
distribution, which enables us to realize secure long-distance communication via short-distance quantum
key distribution.
\end{abstract}


\section{Introduction}
Secure information transmission over a network is an important topic
because there is a risk that a part of the network is broken, contaminated, and/or 
eavesdropped in a large network system.
Cai and Yeung \cite{Cai2002} started the study of 
secure network coding, which offers a method securely transmitting information from the authorized sender to the authorized receiver.
The demand for such a type of secure communication is increasing beyond the areas of information theory and communication theory. 
For example, to realize rigorous security, quantum key distribution is actively studied \cite{BB84}.
Its commercial use has been well developed for limited transmission distance \cite{IDQ}.
However, it is very difficult to directly connect two distinct parties over long distances via quantum key distribution.
To realize long-distance communication with quantum key distribution over short distances, 
they often consider using quantum repeater \cite{BDCZ,DLCZ,SSRG}.
However, it is also difficult to realize quantum repeater.
Hence, it is natural to establish a network whose edges are realized by 
secure communication by one-time-pad use of keys generated  by 
quantum key distribution over short distances.
That is, we generate many pairs of shared secure keys on intermediate nodes by quantum key distribution, where each pair is composed of two nodes 
connected by an edge.
The secure keys shared by two nodes realize a secure channel between the two nodes.
Hence, the application of the network coding to the network composed of these secure channels
yields a scheme to realize secure communication between two distinct parties across a long distance.
However, there is a possibility that a part of the nodes are occupied by Eve.
Hence, we need to use secure network coding instead of network coding.
To resolve this problem,
the pioneering papers \cite{LWLZ,WHZ} studied this type of information transmission for the case of partially trusted routing networks. 
This paper addresses a general network including routing networks, 
in which a part of the nodes are attacked and 
the attacked nodes cannot be identified by the sender and the receiver like the Byzantine setting.
The network topology is not known to the sender and the receiver.

In the above type of network, there are two types of attacks.
In the first type of attack, the malicious adversary, Eve, wiretaps  a subset $E_E$ of all the channels in a network,
which was studied by the first paper by Cai and Yeung \cite{Cai2002}.
Using the universal hashing lemma \cite{bennett95privacy,HILL,hayashi11}, 
the papers \cite{Matsumoto2011,Matsumoto2011a} showed the existence of a secrecy code that works universally for any type of eavesdropper 
when the cardinality of $E_E$ is bounded.
In particular, the paper \cite{KMU1,KMU2,KMU} discussed a concrete construction of such a universally secure code, which is more practical.
As another type of attack on information transmission via a network,
a malicious adversary contaminates the communication by 
changing the information on a subset $E_A$ of all the channels in the network.
Using an error correction, the papers \cite{Cai06a,Cai06,HLKMEK,JLHE} proposed a method to protect the message from contamination.
That is, we require that the authorized receiver correctly recovers the message, which is called robustness.
Now, for simplicity, we consider the unicast setting.
When the transmission rate from the authorized sender, Alice, to the authorized receiver, Bob, is $m_{0}$
and the rate of noise injected by Eve is $m_{1}$,
using the results published in \cite{JLKHKM,Jaggi2008},
the study \cite{JL} showed that
there exists a sequence of asymptotically correctable codes with the rate
$m_{0}-m_{1}$ if the rate of information leakage to Eve is less than $m_{0}-m_{1}$.

However, there is a possibility that the malicious adversary combines eavesdropping and contamination.
That is, contaminating a part of the channels, the malicious adversary might improve the ability of eavesdropping
while a parallel network offers no such a possibility \cite{KSZBJ,ZKBJS1,ZKBJS2}.
In fact, in arbitrarily varying channel model,
noise injection is allowed after Eve's eavesdropping, but  
Eve does not eavesdrop the channel after Eve's noise injection \cite{BBT,Ahlswede,CN88,TJBK}\cite[Table I]{KJL}.
The studies \cite{KMU1,KMU2,KMU,Zhang} discussed the secrecy
when Eve eavesdrops the information transmitted on the channels in $E_E$ after noises are injected in $E_A$,
but they assume that Eve does not know the information of the injected noise. 

In contrast, this paper discusses the secrecy when Eve 
adds artificial information to the information transmitted on the channels in $E_A$, 
eavesdrops the information transmitted on the channels in $E_E$, and 
estimates the original message from the eavesdropped information and the information of the injected noises.
We call this type of attack an active attack 
and call an attack without contamination a passive attack.
Specially, we call each of Eve's active operations a strategy.
Indeed, while the paper \cite{Yao2014} discusses robustness for an active attack,
it discusses secrecy only for a passive attack. 
When $E_A=E_E$ and any active attack is available for Eve, she is allowed to arbitrarily modify the information on the channels in $E_A$ sequentially based on the obtained information.
Fortunately, the previous paper \cite{HOKC1}
showed that an active attack has the same performance as the passive attack in the case of linear codes.

The aim of this paper is as follows.
First, we give a formulation of our setting with
a general set $E_E$ of eavesdropping nodes
and a general set $E_A$ of noise-injection nodes.
Then, we state that no strategy can improve Eve's information
when every operation in the network is linear,
which is a brief review of the result of \cite{HOKC1}.
Next, we discuss a code that satisfies the need for  secrecy and robustness
when the transmission rate from Alice to Bob is $m_{0}$,
the rate of noise injected by Eve is $m_{1}$,
and the rate of information leakage to Eve is $m_{2}$.
In the asymptotic setting, given linear operations on the intermediate nodes with a certain order,
our protocol controls only the encoder in the sender and the decoder in the receiver.
Since we address the static Byzantine setting,
the sender and the receiver do not know the two sets $E_E$ and $E_A$, 
i.e., what edges are attacked,
while they know the channel parameter $m_0$ and the upper bounds of $m_1$ and $m_2$.
Although intermediate nodes make linear operation over a single transmission,
the sender and the decoder are allowed to make coding operations across several transmissions.
Such a coding is called a non-local code to distinguish operations over a single transmission.
We show the existence of such a secure protocol with rate $m_{0}-m_{1}-m_{2}$.
Since our non-local code depends only on the number $m_{0}$, 
the upper bounds of $m_{1}$ and $m_{2}$,
and the dimensions of the input and the output,
it works universally.
Also, we discuss the asymptotic performance when only secrecy is considered.
When Alice and Bob share a small number of initial secret keys 
and can communicate with each other via a public channel, 
we do not need to impose robustness, but need the correctness only for passive attack
because we can make error verification \cite[Section VIII]{Fung}.
In such a case, we show the existence of a secure protocol with the rate $m_{0}-m_{2}$.

Then, we apply these two types of asymptotic analyses to network quantum key distribution. 
Since public communication is allowed in quantum key distribution, 
the latter setting is more useful for the security analysis in
network quantum key distribution.
Further, although the above discussion addresses the unicast setting, 
we explain how to extend this setting to the multiple multicast case including the multiple unicast case.

The remaining part of this paper is organized as follows.
Section \ref{S2} formulates our problem and shows the 
impossibility of Eve's eavesdropping under a linear network,
which is a brief review of the result by \cite{HOKC1}.
Section \ref{S3} discusses the asymptotic setting,
and show the achievability of the asymptotic rate $m_{0}-m_{1}-m_{2}$.
Section \ref{S5} discusses the asymptotic setting with secrecy without robustness.
Using the result of Section \ref{S5},
Section \ref{S6} considers the application of obtained results to 
network quantum key distribution
Section \ref{S10} applies the obtained result to multiple multicast networks.
In Section \ref{SCon}, we state the conclusion.


\section{Secrecy in finite-length setting}\Label{S2}
Although the paper \cite{HOKC1} discusses 
the formulation of a channel model of network 
with eavesdropping and contamination more rigorously,
this section briefly explains this model in the case of an acyclic network with well synchronization.
We consider the unicast setting.
Assume that the authorized sender, Alice, and the authorized receiver, Bob,
are linked via an acyclic network with the set of edges $E$, where the operations on all nodes are linear on the finite field $\FF_q$ with prime power $q$.
Alice inputs the input variable ${\bf X}$\footnote{In this paper, we denote the vector on $\FF_q$ by a bold letter.
But, we use a non-bold letter to describe a scalar and a matrix.}
 in $\FF_q^{m_{{3}}}$ and Bob receives the output variable ${\bf Y}_B$ in $\FF_q^{m_{{4}}}$.
We also assume that the malicious adversary, Eve, wiretaps the information ${\bf Y}_E$ in 
$\FF_q^{m_{6}}$
on the edges of a subset $E_E\subset E$.
Now, we fix the topology and dynamics (operations on the intermediate nodes) of the network.
When we assume all operations on the intermediate nodes are linear,
there exist matrices $K_B \in \FF_q^{m_{{4}}\times m_{{3}} }$ 
and $K_E \in \FF_q^{m_{6}\times m_{{3}} }$
such that the variables ${\bf X}$, ${\bf Y}_B$, and ${\bf Y}_E$ satisfy their relations
\begin{align}
{\bf Y}_B=K_B {\bf X}, \quad
{\bf Y}_E=K_E {\bf X}.\Label{E2}
\end{align}
That is, the matrices $K_B$ and $K_E$ are decided from the network topology and dynamics.
We call this attack the {\it passive attack}.


To address the active attack, we consider stronger Eve, i.e., we assume that Eve adds an error ${\bf Z}$ in $\FF_q^{m_{1}}$  on the edges of a subset $E_A\subset E$.
Using matrices $H_B \in \FF_q^{m_{{4}}\times m_{{1}} }$ 
and $H_E \in \FF_q^{m_{6}\times m_{1}}$, 
we rewrite the above relations as
\begin{align}
{\bf Y}_B=K_B {\bf X}+ H_B {\bf Z}, \quad
{\bf Y}_E=K_E {\bf X}+ H_E {\bf Z}, \Label{E4}
\end{align}
which is called the {\it wiretap and addition model}.
We set the parameter $m_0$ as 
\begin{align}
\rank K_B= m_{0} ,
\Label{1-6}
\end{align}
and assume the ranks of $H_B$ and $K_E$ as
\begin{align}
\rank H_B= m_{1}, ~ \rank K_E= m_{2} .
\Label{1-6X}
\end{align}
Hence, the channel parameters are summarized as Table \ref{hikaku}.

\begin{table}[htpb]
  \caption{Channel parameters}
\Label{hikaku}
\begin{center}
  \begin{tabular}{|c|l|} 
\hline
\multirow{2}{*}{$m_0$} & Rank of the channel from Alice\\
&  to Bob, i.e., $\rank K_B$ \\
\hline
$m_1$ & Rank of Eve's injected information ${\bf Z}$ ($\rank H_B$)\\
\hline
$m_2$ & Rank of Eve's wiretapped information ${\bf Y}_E$ ($\rank K_E$)\\
\hline
$m_3$ & Dimension of Alice's input information ${\bf X}$ \\
\hline
$m_4$ & Dimension of Bob's observed information ${\bf Y}_B$ \\
\hline
$m_5$ & Dimension of Eve's injected information ${\bf Z}$\\
\hline
$m_6$ & Dimension of Eve's wiretapped information ${\bf Y}_E$\\
\hline
  \end{tabular}
\end{center}
\end{table}

Now, to consider the time ordering among the edges in $E$,
we assign the integers to the edges in $E$ such that $E=\{e(1), \ldots, e(k)\}$.
We assume that the information transmission on each edge is done with this order.
In this representation, the elements of $Z$ and $Y_E$ are arranged in this order.
Hence, the elements of the subsets $E_E$ and $E_A$ are expressed as
$E_E=\{e(\zeta(1)), \ldots, e(\zeta(m_{2}))\}$ and
$E_A=\{e(\eta(1)), \ldots, e(\eta(m_{1}))\}$ 
by using two strictly increasing functions $\zeta$ and $\eta$. 
The causality yields that
\begin{align}
H_{E;j,i}=0 \hbox{ when } \eta(i) \ge \zeta(j). \Label{E14}
\end{align}
It is natural that Eve can choose the information to be added in the edge $e(i) \in E_A$
based on the information obtained previously on the edges in the subset 
$\{e(\zeta(j)) \in E_E|\zeta( j) \le \eta(i)\}$.
That is, the added error $Z$ is given as a function $ \alpha$ of $ Y_E$,
which can be regarded as Eve's strategy. 
We call this attack the {\it active attack} with the strategy $\alpha$.

Now, we consider the $n$-transmission setting, in which, Alice uses the same network $n$ times to send the message to Bob.
Alice's input variable (Eve's added variable) is given as 
a matrix $X^n=({\bf X}_1, \ldots, {\bf X}_n) \in\FF_q^{m_{{3}} \times n}$ (a matrix $Z^n=({\bf Z}_1, \ldots, {\bf Z}_n) \in\FF_q^{m_{1} \times n}$),
and Bob's (Eve's) received variable is given as 
a matrix $Y_B^n\in\FF_q^{m_{{4}} \times n}$ (a matrix $Y_E^n\in\FF_q^{m_{2} \times n}$).
We assume that the topology and dynamics of the network 
and the edge attacked by Eve
are not changed during $n$ transmissions.
In particular, the operations on intermediate nodes are not changed.
Then, their relation is given as
\begin{align}
Y_B^n&=K_B X^n+ H_B Z^n, \Label{F1n} \\
Y_E^n&=K_E X^n+ H_E Z^n. \Label{F2n}
\end{align}
Notice that the relations \eqref{F1n} and \eqref{F2n} with $H_E=0$
(only the relation \eqref{F1n}) were treated as the starting point of the paper
\cite{Yao2014} (the papers \cite{JLKHKM,Jaggi2008,JL}).

To discuss the secrecy and the robustness, we formulate a code.
Let ${\cal M}$ and ${\cal L}$ be the message set and the set of values of the scramble random number. 
Then, an encoder is given as a function $\phi_n$ from ${\cal M} \times {\cal L} $
to $\FF_q^{m_{{3}} \times n}$, and the decoder is given as $\psi_n$ from $\FF_q^{m_{{4}} \times n}$ to ${\cal M}$.
Our code is the pair $(\phi_n,\psi_n)$, and is denoted by $\Phi_n$.
Since the code pair of $\phi_n$ and $\psi_n$ is across $n$ transmission,
it is called a non-local code to distinguish operations over a single transmission.
Then, we denote the message and  the scramble random number by $M$ and $L$. 
The cardinality of ${\cal M}$ is called the size of the code and is denoted by $|\Phi_n|$.

Here, we treat $K_B,K_E,H_B,H_E$ as deterministic values, and denote the pairs $(K_B,K_E)$ and $(H_B,H_E)$ by $\bm{K}$ and $\bm{H}$, respectively.
In the following, we fix $\Phi_n,\bm{K},\bm{H},\alpha$.
As a measure of the leaked information, we adopt
the mutual information $I(M; Y_E^n,Z^n)$ between $M$ and Eve's information $Y_E^n$ and $Z^n$
Since the variable $Z^n$ is given as a function of $Y_E^n$,
we have $I(M; Y_E^n,Z^n)=I(M; Y_E^n)$.
Since the leaked information is given as a function of 
$\Phi_n,\bm{K},\bm{H},\alpha$ in this situation, 
we denote it by $I(M;Y_E^n)[\Phi_n,\bm{K},\bm{H},\alpha]$.
When we always choose $Z^n=0$, the attack is the same as the passive attack.
This strategy is denoted by $0$.
and chooses her strategy dependently of the position.

Now, we have the following theorem \cite{HOKC1}.
\begin{theorem}\Label{T1}
For any Eve's strategy $\alpha$, 
Eve's information $Y_E^n$ with strategy $\alpha$
and that with strategy $0$ can be simulated by each other.
Hence, we have the equation
\begin{align}
I(M;Y_E^n)[\Phi_n,\bm{K},\bm{H},0]
=
I(M;Y_E^n)[\Phi_n,\bm{K},\bm{H},\alpha].
\end{align}
\end{theorem}

While the paper \cite{HOKC1} showed this theorem in a more general setting
with more rigorous description of the problem setting,
we briefly give its proof in Appendix \ref{ApB}.
This theorem shows that the information leakage of the active attack with the strategy $\alpha$ is the same as 
the information leakage of the passive attack.
Hence, to guarantee the secrecy under an arbitrary active attack,
it is sufficient to show the secrecy under the passive attack.

\begin{remark}[Number of choices]
To compare passive and active attacks,
we count the number of choices of both attacks.
While the passive attack is characterized by the matrix $K_E$,
the information leaked to Eve in the passive attack depends only on the kernel of the matrix $K_E$.
To characterize the information leaked to Eve,
we consider two matrices to be equivalent when their kernels are the same.
In a passive attack, when we fix the rank of $K_E$ (the dimension of leaked information),
by taking into account the equivalent class,
the number of possible choices is upper bounded by
$q^{m_{2}(m_{{3}}-m_{2})} $. 
With an active attack, this calculation is more complicated.
For simplicity, we consider the case with $n=1$.
To consider the minimum number of choices of $\alpha^n$,
we assume condition \eqref{E14}.
(When we do not make this assumption, the number of choices is larger.)
We do not count the choice for the inputs on the edge $\eta(i)$ with $\eta(i) \ge \zeta(m_{6})$ because it does not affect Eve's information.
Then, even when we fix the matrices $K_E,H_E$,
the number of choices of $\alpha^n$ is 
\begin{align}
q^{\sum_{i: \eta(i) < \zeta(m_{6})} q^{T_i}}, \Label{F17}
\end{align}
where $T_i:= \max \{j | \eta(i) \ge \zeta(j)\}$.
Notice that $T_i=i$ when $E_A=E_E$.
If we count the choice on the remaining edges, we need to multiply $q^{\sum_{i: \eta(i) \ge \zeta(m_{6})} q^{T_i}}$ on \eqref{F17}.
For a generic natural number $n$,
the number of choices of $\alpha^n$ is
\begin{align}
 q^{n\sum_{i: \eta(i) < \zeta(m_{6})} q^{n T_i}} \Label{F18}.
\end{align}
\end{remark}

\section{Asymptotic setting with secrecy and robustness}\Label{S3}
In this section, under the model given in Section \ref{S2},
we consider the asymptotic setting by taking account of robustness as well as secrecy
while Eve's strategy $\alpha$ is assumed to satisfy the uniqueness condition.
We previously assumed that 
the matrices $K_B$, $K_E$, $H_B$, and $H_E$, i.e., 
the topology and dynamics of the network and the edge attacked by Eve do not change during $n$ transmissions.
Now, we assume that Eve knows these matrices and 
that Alice and Bob know none of them because Alice and Bob often do not know the topology and/nor dynamics of the network and/nor the places of the edges attacked by Eve.
However, due to the limitation of Eve's ability,
we assume that the dimension of the information leaked to Eve and 
the rank of the information injected by Eve are limited to $m_2 $ and $m_1$, respectively.
Indeed, when the original network is given by the graph $(V,E)$ and
Eve eavesdrops at most $m_6'$ edges and injects the noise at most $m_5'$ edges,
we have $m_2 \le m_6'$ and $m_1\le m_5'$.
This evaluation is still valid even in the wiretap and replacement model.
Therefore, it is natural to assume the upper bounds of these dimensions.
(See Remark \ref{R-Dim}.)

When Eve adds the error $Z^n$, there is a possibility that Bob cannot recover the original information $M$.
This problem is called {\it robustness}, and may be regarded as a kind of error correction.
Under the conventional error correction, 
the error $Z^n$ is treated as a random variable subject to a certain distribution.
However, our problem 
is different from the conventional error correction 
because the decoding error probability depends on the strategy $\alpha^n$.
Hence, we denote it by $P_e[\Phi_n,\bm{K},\bm{H},\alpha^n]$.
Then, the following proposition is known.

\begin{proposition}[\protect{\cite{JLKHKM,Jaggi2008,JL,Yao2014}}]\Label{T2}
Assume that 
$m_{2}+m_{1}< m_{0}$.
There exists a sequence of non-local codes $\Phi_{n}$ of block-length $l_n$ on a finite field $\FF_q$
whose message set is $\FF_{q}^{k_{n}}$ such that
\begin{align}
&\lim_{n \to \infty} \frac{k_n}{l_n} = m_{0}-m_{1}\\
&\lim_{n \to \infty} \max_{ \bm{K},\bm{H}} \max_{\alpha^n} 
P_e[\Phi_n, \bm{K},\bm{H},\alpha^n]=0,
\end{align}
where the maximum is taken with respect to
$(K_B,H_B,K_E,H_E) \in 
\FF_q^{m_{{4}} \times m_{{3}}}\times 
\FF_q^{m_{{4}} \times m_{5}}\times 
\FF_q^{m_{6} \times m_{{3}}}\times 
\FF_q^{m_{6} \times m_{5}}$ with conditions \eqref{1-6} and 
\eqref{1-6X}.
Here, there is no restriction for the choice of $m_5$ and $m_6$.
\end{proposition}

The existing proof of Proposition \ref{T2}
is given as a combination of several results.
Each part of the existing proof is hard to read because it omits the detail derivation. 
Hence, for readers' convenience, 
we give its alternative proof in Appendix \ref{Ap1}, which has an improvement over the existing proof. 
Combining Theorem \ref{T1} and Proposition \ref{T2}, we obtain the following theorem.
\begin{theorem}\Label{T3}
We assume that
$m_{2}+m_{1}< m_{0}$.
There exists a sequence of non-local codes $\Phi_{n}$ 
of block-length $l_n$ on finite field $\FF_q$
whose message set is $\FF_q^{k_n}$ such that
\begin{align}
&\lim_{n \to \infty} \frac{k_n}{l_n} = m_{0}-m_{1}-m_{2}\\
&\lim_{n \to \infty} \max_{ \bm{K},\bm{H}} \max_{\alpha^n} 
P_e[\Phi_n,\bm{K},\bm{H},\alpha^n]=0 \Label{H3-181} \\
&
\max_{ \bm{K},\bm{H}} \max_{\alpha^n} 
I(M;Y_E^n)[\Phi_n,\bm{K},\bm{H},\alpha^n]=0, \Label{H3-182}
\end{align}
where the maximum is taken in the same way as with Proposition \ref{T2}.
\end{theorem}

Before our proof, we prepare basic facts about information-theoretic security. 
We focus on a random hash function $f_R$ from ${\cal X}$ to ${\cal Y}$ with random variable $R$ deciding the function $f_R$.
It is called {\it universal2} when 
\begin{align}
{\rm Pr}\{ f_R(x)=y\}\le \frac{|{\cal Y}|}{|{\cal X}|}
\end{align}
for any $x \in {\cal X}$ and $y \in {\cal Y}$.

For $s\in (0,1]$, we define the conditional R\'{e}nyi entropy $H_{1+s}(X|Z)$ for the joint distribution $P_{XZ}$ as \cite{hayashi11}
\begin{align}
H_{1+s}(X|Z):= \frac{-1}{s}\log \sum_{z \in {\cal Z}} P_Z(z) \sum_{x \in {\cal X}} P_{X|Z}(x|z)^{1+s},
\end{align}
which is often denoted by $H_{1+s}^{\uparrow}(X|Z)$ in \cite{TBH,HW}.
When $X$ obeys the uniform distribution, we have 
\begin{align}
H_{1+s}(X|Z) \ge  \log\frac{  |{\cal X}|}{|{\cal Z}|}. \Label{F15}
\end{align}

\begin{proposition}
\cite{bennett95privacy,HILL}\cite[Theorem 1]{hayashi11}\Label{T4}
\begin{align}
I(f_R(X);Z|R) \le \frac{e^{s \log |{\cal Y}|- H_{1+s}(X|Z)}}{s}
\end{align}
\end{proposition}
for $s \in (0,1]$.

\begin{proofof}{Theorem \ref{T3}}
We choose a sequence of non-local codes $\{\Phi_n=(\phi_n,\psi_n)\}$  given in Corollary \ref{T2}.
We fix $ \bar{k}_n:= k_n- m_{2} {l_n}- \lceil \sqrt{l_n}\rceil$.
Now, we choose a universal2 linear surjective random hash function $f_R$ from $\FF_q^{k_n}$ to $\FF_q^{\bar{k}_n}$.

To construct our non-local code, we consider a virtual protocol as follows. 
First, Alice sends a larger message $M$ by using the non-local code $\Phi_n$,
and Bob recovers it.
Second, Alice randomly chooses $R$ deciding the hash function $f_R$ and sends it to Bob via a public channel.
Finally, Alice and Bob apply the hash function $f_R$ to their message, and denote the  result value by $\bar{M}$ so that Alice and Bob share the information  $\bar{M}$ with a probability of close to $1$.

Since the conditional mutual information between
$\bar{M}$ and $Y_E^{l_n}$ depends on $\Phi_n,\bm{K},\bm{H},\alpha^n$,
we denote it by $ I(\bar{M};Y_E^{l_n}|R)[\Phi_n,\bm{K},\bm{H},\alpha^n]$.
Theorem \ref{T1} shows $ I(\bar{M};Y_E^{l_n}|R)[\Phi_n,\bm{K},\bm{H},\alpha^n]=
 I(\bar{M};Y_E^{l_n}|R)[\Phi_n,\bm{K},\bm{H},0]$, which does not depend on $K_B,H_B ,H_E$ and depends only on $K_E$.
Now, we evaluate this leaked information via a similar idea to that reported in \cite{Matsumoto2011}.
Since inequality \eqref{F15} implies that 
$H_{1+s}(M|Y_E^{l_n}) \ge( k_n -{l_n} m_{2}) \log q$,
Proposition \ref{T4} yields
\begin{align}
I(\bar{M};Y_E^{l_n}|R)[\Phi_n,\bm{K},\bm{H},0]
\le &\frac{e^{s (\bar{k}_n\log q - H_{1+s}(M|Y_E^{l_n}))}}{s} \nonumber \\
\le & \frac{q^{s (\bar{k}_n- k_n +{l_n} m_{2})}}{s}
\le \frac{q^{-s \lceil \sqrt{l_n}\rceil }}{s}.
\end{align}
We set $s=1$.
For each matrix $K_E \in \FF_q^{m_{6}\times m_{{3}}}$ satisfying $\rank K_E= m_2$,
Markov inequality guarantees that the inequality
\begin{align}
I(\bar{M};Y_E^{l_n})[\Phi_n,\bm{K},\bm{H},0]|_{R=r}
\le q^{- \lceil \sqrt{l_n}\rceil +c+1} \Label{F20a}
\end{align}
holds at least with probability $1-q^{-c-1}$.
Since the number of matrices $K_E$ satisfying $\rank K_E=m_{2} $ is upper bounded by $ q^{m_{6}m_{{3}}}$,
there exists a matrix $K_E \in \FF_q^{m_{6}\times m_{{3}}}$ such that
$\rank K_E= m_2$ and \eqref{F20a} does not hold
at most with probability $q^{m_{6}m_{{3}}} q^{-c-1}$.
Hence, \eqref{F20a} holds for any matrix $K_E \in \FF_q^{m_{6}\times m_{{3}}}$ satisfying $\rank K_E= m_2$
at least with probability $1-q^{m_{6}m_{{3}}} q^{-c-1}$.
Letting $c$ be $m_{6}m_{3}$, we have
\begin{align}
I(\bar{M};Y_E^{l_n})[\Phi_n,\bm{K},\bm{H},0]|_{R=r}
\le q^{- \lceil \sqrt{l_n}\rceil +m_{6}m_{{3}}+1} \Label{F20}
\end{align}
for any matrix $K_E \in \FF_q^{m_{6}\times m_{{3}}}$ satisfying $\rank K_E= m_2$
at least with probability $1-\frac{1}{q}$.
Therefore, there exists a suitable hash function $f_r$ such that
\begin{align*}
I(\bar{M};Y_E^{l_n})[\Phi_n,\bm{K},\bm{H},0]|_{R=r}
\le q^{- \lceil \sqrt{l_n}\rceil +m_{6}m_{{3}}+1},
\end{align*}
which goes to zero as $n$ goes to infinity because 
$m_{6}m_{{3}}+1$ is a constant.
Since the code on our network is linear,
Eve observes a subspace of input information
$\FF_q^{\bar{k}_n} $.
Hence, the amount of leaked information is an integer times $\log q$.
Hence, as discussed in  \cite{Matsumoto2011a},
when $l_n$ is sufficiently large, there exists a suitable hash function $f_r$ such that
\begin{align*}
I(\bar{M};Y_E^{l_n})[\Phi_n,\bm{K},\bm{H},0]|_{R=r}
=0.
\end{align*}

Now, we return to the construction of non-local codes.
We choose the sets ${\cal M}$ and ${\cal L}$ as
$\FF_q^{\bar{k}_n}$ and $\FF_q^{m_{2} {l_n}+ \lceil \sqrt{l_n}\rceil}$, respectively.
Since the linearity and the surjectivity of $f_r$ implies that
$|f_r^{-1}(x)|=q^{m_{2} {l_n}+ \lceil \sqrt{l_n}\rceil}$ for any element $x\in {\cal M} $,
we can define the invertible function $\bar{f}_r$ from ${\cal M}\times {\cal L}$
to the domain of $f_r$, i.e., $\FF_q^{k_n} $ such that $\bar{f}_r^{-1} (f_r^{-1}(x))=\{x\}\times {\cal L}$
for any element $x\in {\cal M} $.
This condition implies that $f_r \circ \bar{f}_r(x,y)=x$ for $(x,y)\in {\cal M}\times {\cal L}$.
Then, we define our non-local encoder as $\bar{\phi}_n:= \phi_n \circ \bar{f}_r$,
and our non-local decoder as $\bar{\psi}_n:= f_r\circ \psi_n $.
The sequence of non-local codes $\{(\bar{\phi}_n,\bar{\psi}_n)\}$
satisfies the desired requirements.
\end{proofof}

\begin{remark}\Label{R-Dim}
If we replace the condition \eqref{1-6X} by the condition
\begin{align}
\rank H_B\le m_{1}, ~ \rank K_E\le m_{2} ,
\Label{1-6X2}
\end{align}
the Proposition \ref{T2} and Theorem \ref{T3} still hold due to the following reason.
For $(K_B,H_B,K_E,H_E)$ to satisfy \eqref{1-6} and \eqref{1-6X2}, 
there exists $(K_B',H_B',K_E',H_E')$ to satisfy \eqref{1-6} and \eqref{1-6X} such that 
$P_e[\Phi_n, \bm{K},\bm{H},\alpha^n] \le P_e[\Phi_n, \bm{K}',\bm{H}',\alpha^n]$
and
$I(M;Y_E^n)[\Phi_n,\bm{K},\bm{H},\alpha^n]\le
I(M;Y_E^n)[\Phi_n,\bm{K}',\bm{H}',\alpha^n]$.
Hence,
the Proposition \ref{T2} and Theorem \ref{T3} still hold under this modification.
\end{remark}

\begin{remark}[Efficient non-local code construction]
We discuss an efficient construction of our non-local code from a non-local code 
$(\phi_n,\psi_n)$ given in Corollary \ref{T2} with $q=2$.
A modified form of the Toeplitz matrices is also shown to be 
a universal2 linear surjective hash function, 
which is given by a concatenation $(T(S), I)$ of the 
$(m_{2} l_n+ \lceil \sqrt{l_n}\rceil) \times \bar{k}_n$ Toeplitz matrix $T(S)$ and 
the $\bar{k}_n \times \bar{k}_n$ identity matrix $I$ \cite{Hayashi-Tsurumaru},
where $S$ is the random seed used to decide the Toeplitz matrix and belongs to $\FF_2^{k_n-1}$.
The (modified) Toeplitz matrices are particularly useful in practice, because there exists an efficient multiplication algorithm using the fast Fourier transform algorithm with complexity $O(l_n\log l_n)$. 

When the random seed $S$ is fixed, 
the encoder for our non-local code is given as follows.
By using the scramble random variable $L \in \FF_2^{m_{2} l_n+ \lceil \sqrt{l_n}\rceil)}$,
the encoder $\bar{\phi}_n$ is given as $\phi_{n}
\Big( 
\Big(
\begin{array}{cc} 
I & - T(S) \\
0 & I
\end{array}
\Big)
\Big(
\begin{array}{c} 
M \\
L
\end{array}
\Big)
\Big)$
because 
$
(I,T(S))
\Big(
\begin{array}{cc} 
I & - T(S) \\
0 & I
\end{array}
\Big)
= (I, 0)$.
(The multiplication of Toeplitz matrix $T(S)$ can be performed as a part of a circulant matrix. 
For example, the reference \cite[Appendix C]{Hayashi-Tsurumaru}
provides a method to give a circulant matrix.).
A more efficient construction for univeral2 hash function is discussed in \cite{Hayashi-Tsurumaru}.
Hence, the decoder $\bar{\psi}_n$ is given as $Y_B^{l_n} \mapsto (I,T(S)) \psi_{n}(Y_B^{l_n})$.
\end{remark}

\begin{remark}
Here, we clarify the difference between our results and the setting of the preceding papers \cite{KMU1,KMU2,KMU,Yao2014,Zhang,Rashmi}, which consider correctness and secrecy.
Their secrecy analysis is different from our analysis
although the non-local code construction in \cite{KMU1,KMU2,KMU,Yao2014,Zhang} does not depend on the concrete form of 
matrices $K_B, K_E, H_B, H_E$, which is similar to our non-local code construction.

While the papers \cite{SK,Yao2014} considered correctness when the error exists,
it discusses the secrecy only when there is no error.
Similarly, the paper \cite{Rashmi} considers a different active adversary model, in which,
it discusses
the node-repair and data-reconstruction operations
even in the presence of such an attack
while the model of passive eavesdroppers in the paper \cite{Rashmi} 
discusses the secrecy with respect to the message to be transmitted.
Indeed, the papers \cite{SK,Yao2014} provided a statement similar to Theorem \ref{T3}.
However, it showed only Eq. \eqref{H3-181} and 
$\lim_{n \to \infty} \max_{ \bm{K},\bm{H}} 
I(M;Y_E^n)[\Phi_n,\bm{K},\bm{H},0]=0$
instead of \eqref{H3-182}
by combining Proposition \ref{T2B} and the result of the paper \cite{SK}. 
To show \eqref{H3-182}, we need to employ Theorem \ref{T1}.
If we do not apply Theorem \ref{T1} in step \eqref{F20} in our proof of Theorem \ref{T3},
we have to multiply the number of choices of strategy $\alpha^n$.
As a generalization of \eqref{F17}, this number is given in \eqref{F18},
which grows up double-exponentially. 
Hence, our proof of Theorem \ref{T3} does not work without the use of Theorem \ref{T1}.

While the papers \cite[Proposition 5]{KMU}\cite{Zhang} consider the secrecy when the error exists,
it addresses the amount of leaked information only when the eavesdropper does not know the information of the noise.
That is, they evaluate the mutual information between $ M$ and $Y_E^n $.
However, our analysis evaluates the leaked information when the eavesdropper knows the information about the noise.
That is, we address the mutual information between $ M$ and the pair $(Y_E^n,Z^n) $.
\end{remark}

\section{Asymptotic setting with secrecy}\Label{S5}
When Alice and Bob can communicate via public channel,
the verification of correctness can be done by a universal2 hash function \cite[Section VIII]{Fung}\cite[Step 4 of Protocol 2]{H17}.
When we employ a modified Toeplitz matrix as a universal2 hash function
and $m_2$ bits are exchanged for the verification,
it has only calculation complexity $O(m_2 \log m_2)$.
Due to this step, 
we can guarantee the correctness with probability $1-2^{-m_2}$, which is called the significance level\cite[Section VIII]{Fung}.
Hence, in this case, we consider the case when only the secrecy is imposed and the robustness is not imposed.
That is, we impose the following condition.
\begin{align}
\lim_{n \to \infty} \max_{ \bm{K}} P_e[\Phi_n, \bm{K},0,0]
=\lim_{n \to \infty} \max_{ \bm{K},\bm{H}} P_e[\Phi_n, \bm{K},\bm{H},0]=0.
\Label{H3-18K}
\end{align}
However, even when the verification of correctness is passed,
there is a possibility that Eve can make an active attack to satisfy 
\begin{align}
P_e[\Phi_n, \bm{K},\bm{H},\alpha^n]=0.
\Label{H3-30}
\end{align}
Hence, as the secrecy, we impose the following condition.
\begin{align}
 \max_{ \bm{K},\bm{H}} \max_{\alpha^n} 
I(M;Y_E^n)[\Phi_n,\bm{K},\bm{H},\alpha^n]=0\Label{47-2} .
\end{align}
Here, both maximums are taken with respect to
$(K_E,H_E) \in 
\FF_q^{m_{{6}} \times m_{3}}\times 
\FF_q^{m_{{6}} \times m_{5}}$ with \eqref{1-6} and \eqref{1-6X}.
We notice that the situation of the correctness \eqref{H3-18K}
is different from the situation of the secrecy \eqref{47-2}.
Indeed, it is possible to restrict the range of $\alpha^n$
by imposing the condition \eqref{H3-30}.
However, since it is not so easy to handle the condition \eqref{H3-30},
the maximum for $\alpha^n$
is addressed without the restriction \eqref{H3-30}.
That is, the correctness \eqref{H3-18K} addresses only the case with passive attack,
but the secrecy \eqref{47-2} addresses the cases with active attack.

This setting appears when we consider quantum key distribution, as explained in Section \ref{S6}.
Then, we have the following theorem to analyze this problem.

\begin{theorem}\Label{T3B}
There exists a sequence of non-local codes $\Phi_{n}$ 
of block-length $l_n$ on finite field $\FF_q$
whose message set is $\FF_q^{k_n}$ such that conditions \eqref{H3-18K} and \eqref{47-2} and
\begin{align}
&\lim_{n \to \infty} \frac{k_n}{l_n} = m_{0}-m_{2} \Label{47-1}
\end{align}
holds.
\end{theorem}

From the definition, we see that
$P_e[\Phi_n, \bm{K},0,0]=P_e[\Phi_n, \bm{K},\bm{H},0]$.
Also, note that $P_e[\Phi_n, \bm{K},0,\alpha^n]$ does not depend on $K_E$.
Further, the rate $m_{0}-m_{2}$ is asymptotically optimal, 
i.e., there is no non-local code surpassing the rate $m_{0}-m_{2}$,
which follows from the converse part of the conventional wire-tap channel \cite{wyner75,csiszar78}.

To show the above theorem, as a special case of Theorem \ref{T3} with $m_1=0$, 
we prepare the following corollary.

\begin{corollary}\Label{T2E}
There exists a sequence of non-local codes $\Phi_{n}$ 
of block-length $l_n$ on finite field $\FF_{q}$
whose message set is $\FF_{q}^{k_{n}}$ such that
\begin{align}
&\lim_{n \to \infty} \frac{k_{n}}{l_n} = m_{0}-m_{2}\\
&
\max_{ \bm{K}}
I(M;Y_E^n)[\Phi_n,\bm{K},0,0]=0,\\
&\lim_{n \to \infty} \max_{ \bm{K}} 
P_e[\Phi_n, \bm{K},0,0]=0,
\end{align}
where the maximum is taken with respect to
$(K_B,K_E) \in 
\FF_q^{m_{{4}} \times m_{{3}}}\times 
\FF_q^{m_{{6}} \times m_{3}}$ under the conditions \eqref{1-6} and \eqref{1-6X}.
\end{corollary}

Combining Corollary \ref{T2E} and Theorem \ref{T1},
we obtain Theorem \ref{T3B}.

Here, we compare existing results with Corollary \ref{T2E}.
As a similar result to Corollary \ref{T2E}, the following proposition is known.
Since Corollary \ref{T2E} does not require the assumptions $m_3=m_4=m_0 $ and $K_B=I$,
Corollary \ref{T2E} is slightly advantageous.
Hence, Theorem \ref{T3B} is a stronger statement than the following existing statement.

\begin{proposition}[\protect{\cite[Theorem 7]{Matsumoto2011a}},\cite{KMU}]\Label{T2D}
We assume that $m_3=m_4=m_0 $ and $K_B=I$.
There exists a sequence of non-local codes $\Phi_{n}$ 
of block-length $n$ on finite field $\FF_{q}$
whose message set is $\FF_{q}^{k_{n}}$ such that
\begin{align}
&\lim_{n \to \infty} \frac{k_{n}}{n} = m_{0}-m_{2}\\
&\lim_{n \to \infty} \max_{ \bm{K}}
I(M;Y_E^n)[\Phi_n,\bm{K},0,0]=0,\\
&\lim_{n \to \infty} 
P_e[\Phi_n, \bm{K},0,0]=0,
\end{align}
where the maximum is taken with respect to
$K_E \in \FF_q^{m_{{6}} \times m_{3}}$ with condition \eqref{1-6X}.
\end{proposition}


\section{Application to network quantum key distribution}\Label{S6}
In this section, to realize long-distance communication with quantum key distribution,
using the result in Sections \ref{S3} and \ref{S5},
we consider a network of quantum key distribution.
Although the existing studies \cite{LWLZ,WHZ} discussed a similar case over routing networks,
they did not discuss the relation with network code including active attacks.

Assume  that the authorized sender, Alice, is connected to the authorized receiver, Bob,
via the network given by the graph $({V}, {E})$ with $|E|=k$.
A linear operation is fixed in each node so that 
we have the relation $Y_B=K_B X$ with Alice's input $X$ and Bob's output $Y_B$.
Then, if secure information transmission is available on each edge,
secure communication from Alice to Bob can be realized.
For every edge $(u,v)\in E$,
the distant nodes $u$ and $v$ generate secure common keys by quantum key distribution.
That is, $k$ pairs of secure keys are generated by quantum key distribution.
In the following, 
we discuss how we can make secure message transmission from Alice to Bob
by using these $k$ pairs of secure keys with public channels.
This kind of secure communication is called network quantum key distribution.

First, we consider the case when all nodes are authenticated.
In this case, Alice can securely send her message $X$ to Bob in the following way.
Let $X_i$ be the random variable to be transmitted on the $i$-th edge.
Let $Z_i$ be the secure keys generated in the $i$-th edge by quantum key distribution.
When $X_i$ is directly transmitted, this information transmission is not secure.
To realize security, $Y_i:=X_i+Z_i$ is transmitted on the $i$-th edge, instead.
Then, a secure transmission in each edge is realized.
Hence, due to the above relation $Y_B=K_B X$, 
secure communication from Alice to Bob can be realized.

However, it is very difficult to guarantee security when a part of the nodes are occupied by Eve.
Such a model is often called a node adversary model  
while the model introduced in Section \ref{S2} is called an edge adversary model.
The main problem with network quantum key distribution is the realization of secure communication from Alice to Bob 
under a node adversary model. 
To investigate the security in the node adversary model, 
we convert a given node adversary model to a special case of the edge adversary model as in \cite{TJBK}.
In an edge adversary model,
Eve wiretaps and contaminates the information only on the edges $E_E$.
To apply the model to the current situation, we consider that 
all the edges linked to the nodes occupied by Eve are wiretapped and controlled by Eve.
When these occupied nodes communicate with each other, Eve's attack is an active attack.
That is, analysis for active attack is essential.
Therefore, we can apply Theorem \ref{T3} 
to the security analysis of 
the direct transmission of the secret message via network quantum key distribution.
In quantum key distribution, it is usual to assume that Alice and Bob share 
secure random numbers whose lengths are asymptotically negligible in block-length $n$
because the asymptotically negligible keys are needed for 
authentication for the public channel.
In this case, to generate secure keys with length $O(n)$, 
we can employ Theorem \ref{T3B},
where the asymptotically negligible keys are used for an error verification test.

\begin{figure}[htbp]
\begin{center}
\includegraphics[scale=0.4]{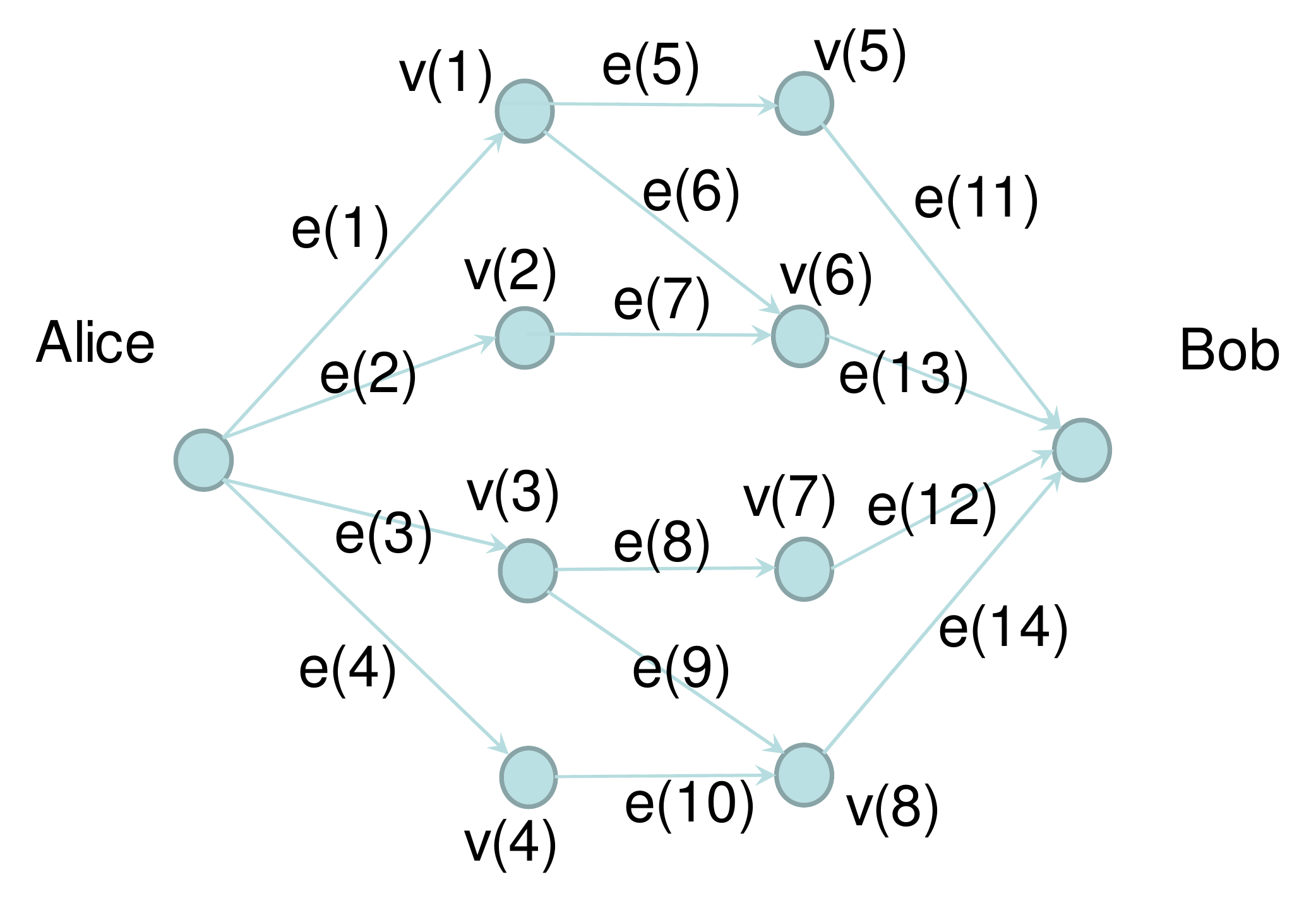}
\end{center}
\caption{Network with name of edges}
\Label{F3}
\end{figure}%

\begin{figure}[htbp]
\begin{center}
\includegraphics[scale=0.4]{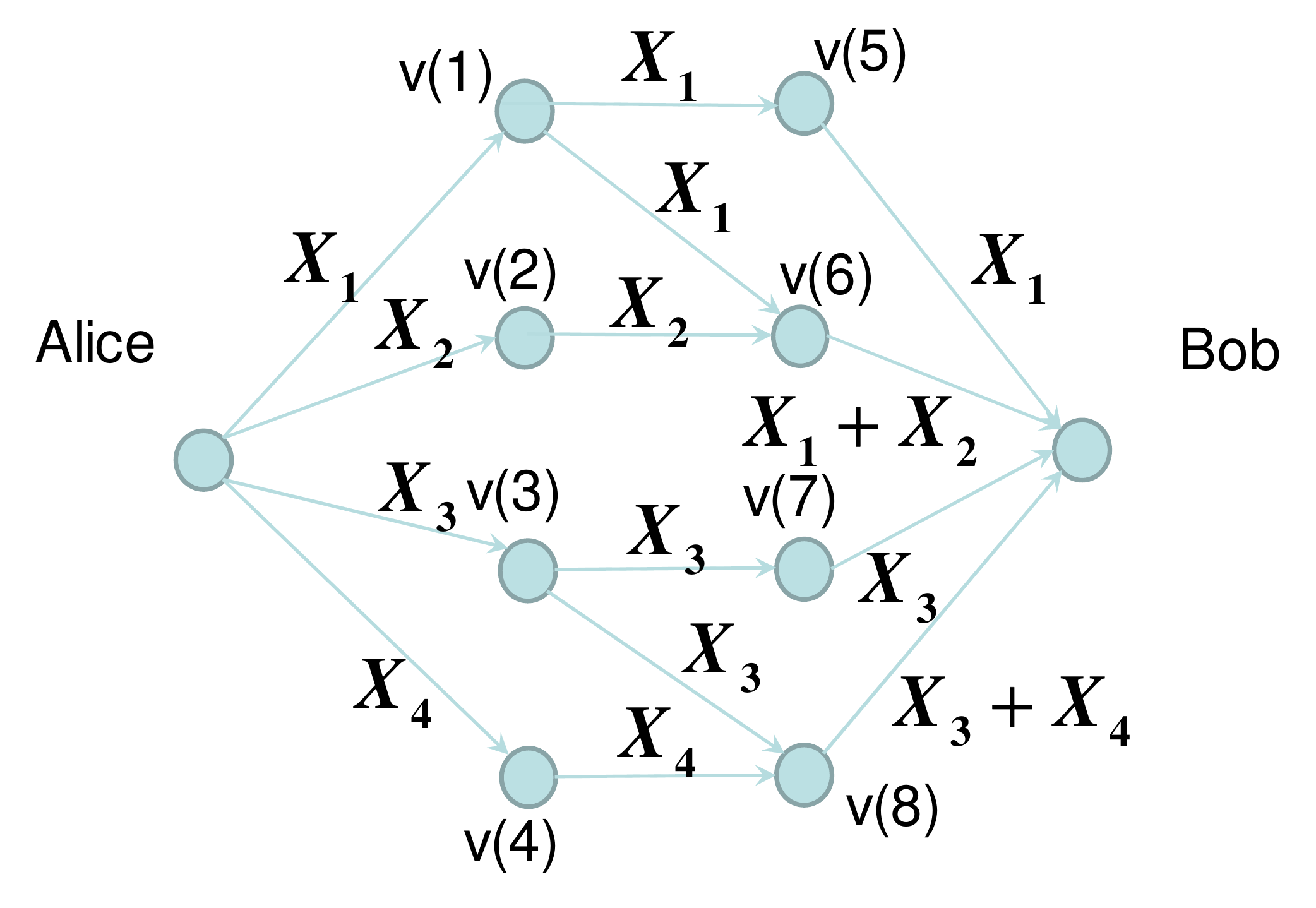}
\end{center}
\caption{Network with network flow}
\Label{F1B}
\end{figure}%

For example, we consider the network given in Fig. \ref{F3}, which has nodes
$v(1), \ldots, v(8)$ as intermediates nodes.
Fig. \ref{F1B} expresses the information flow on each edge in this network.
This network connects Alice and Bob with rank 4.
The ranks of $K_E$ and $H_B$ of typical cases
are summarized in Table \ref{TB2}.
Also, Section II-D of \cite{HOKC1} discusses the same network.

\begin{table}[htb]
  \begin{center}
    \caption{Ranks of $K_E$ and $H_B$ dependently of attacked nodes}\Label{TB2}
    \begin{tabular}{|c|c|c|} \hline
Nodes to be attacked & rank $K_E$ & rank $H_B$ \\ \hline 
$v(1)$ & 1 & 2 \\ \hline 
$v(2)$ & 1 & 1 \\ \hline 
$v(6)$ & 2 & 1 \\ \hline 
$v(2) \& v(5)$ & 2 & 2 \\ \hline 
$v(2) \& v(6)$ & 2 & 1 \\ \hline 
$v(1) \& v(3)$ & 2 & 4 \\ \hline 
$v(1) \& v(2)$ & 2 & 3 \\ \hline 
$v(1) \& v(6)$ & 2 & 2 \\ \hline 
$v(1) \& v(8)$ & 3 & 3 \\ \hline 
$v(6) \& v(8)$ & 4 & 2 \\ \hline 
    \end{tabular}
  \end{center}
\end{table}

When the number of nodes occupied by Eve is limited to 1,
the ranks of $K_E$ and $H_B$ are upper bounded by $2$.
In the latter case, Theorem \ref{T3B} guarantees that
Alice can securely transmit a random number with rank 2 per single use of the network.
In the former case, since $4-2-2=0$, Theorem \ref{T3} cannot guarantee that Alice  securely transmits her message to Bob.

As another example, we consider the circle type network given in Fig. \ref{F4}, in which, 
the nodes connect the next nodes and the nodes after the next.
Assume that we have pairs of secret keys in the circle type network of Fig. \ref{F4}.
We suppose that $v(1)$ intends to communicate with $v(8)$ securely.
They make the network as
$v(1) \to v(12)\to v(10)\to v(8)$,
$v(1) \to v(11)\to v(9)\to v(8)$,
$v(1) \to v(3)\to v(5)\to v(7)\to v(8)$,
$v(1) \to v(2)\to v(4)\to v(6)\to v(8)$, which connects $v(1)$ and $v(8)$ with rank 4.
When Eve occupies one intermediate node,
the ranks of $K_E$ and $H_B$ are one.
In the latter case, Theorem \ref{T3B} guarantees that
Alice in $v(1)$ can securely transmit a random number with rank 3 per single use of the network.
In the former case, 
Theorem \ref{T3} guarantees that Alice in $v(1)$ securely transmits her message to Bob with rank 2 per single use of the network.

\begin{figure}[htbp]
\begin{center}
\includegraphics[scale=0.37]{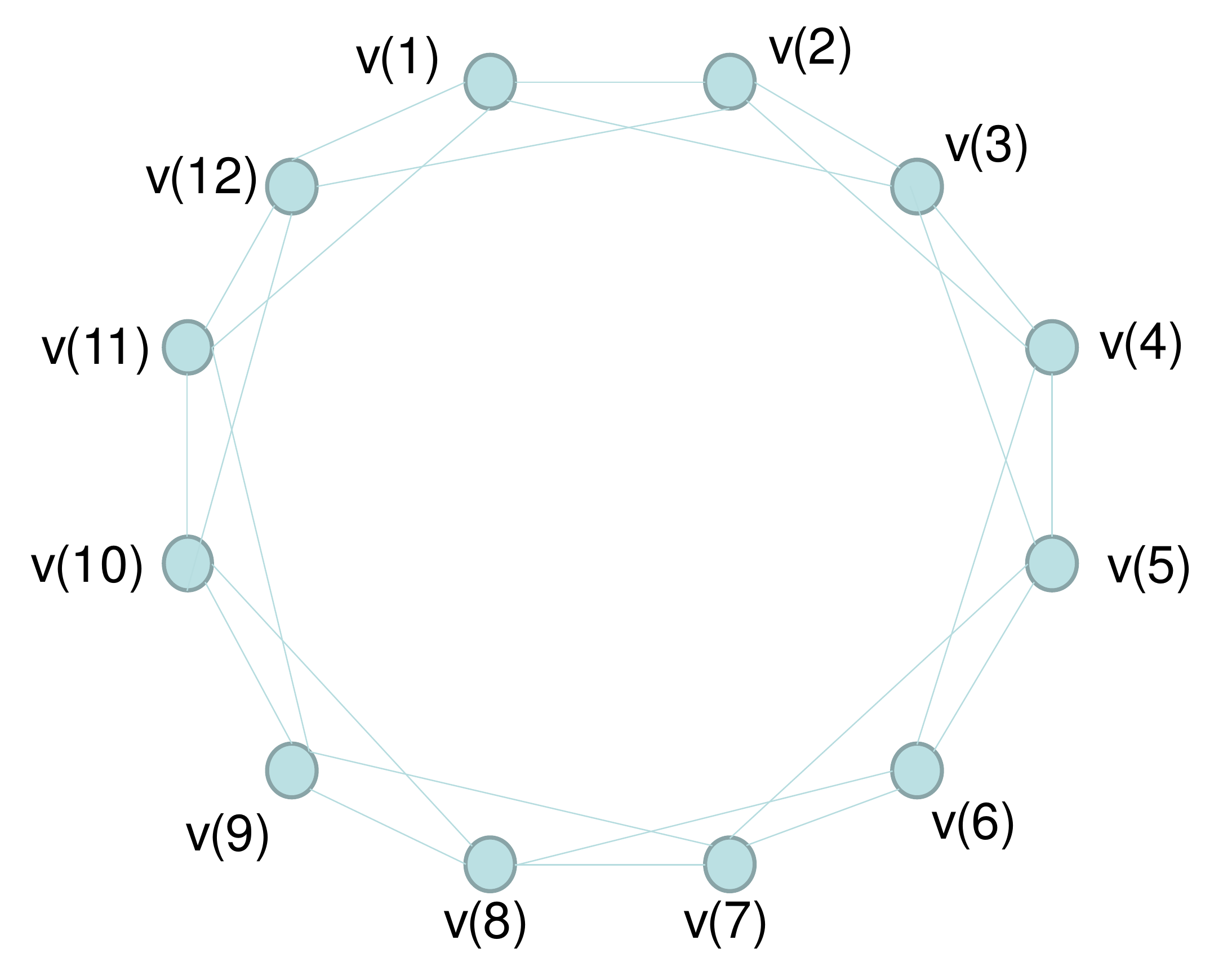}
\end{center}
\caption{Circle type network}
\Label{F4}
\end{figure}%

When Eve occupies two intermediate nodes, the ranks of $K_E$ and $H_B$ are at most two.
In the latter case, Theorem \ref{T3B} guarantees that
Alice in $v(1)$ can securely transmit a random number with rank 2 per single use of the network.
In the former case, 
Theorem \ref{T3} cannot guarantee that Alice in $v(1)$ securely transmits her message to Bob.
This method can be generalized to the case when Alice and Bob are $v(i)$ and $v(j)$
with $|i-j|\ge 2 (\bmod~ 12)$.

Indeed, this idea can be generalized to this circle type network even when the number of nodes is odd. 
Further, this network can be generalized to the following network of quantum key distribution with two integers $k>l>0$.
The set of nodes is given as $\{v(i)\}_{i=1}^k$,
and the set of edges is given as $\{(v(i), v(j)) \}_{|i-j| \le l (\bmod~ k)}$.
Now, we set Alice and Bob as $v(i)$ and $v(j)$ with $|i-j|\ge l (\bmod~ k)$.
Then, they can make $2l$ paths connecting $v(i)$ and $v(j)$ without duplication in the intermediate nodes.
Then, even when $2l-1$ nodes are occupied by Eve,
Alice and Bob can securely share a secret random number due to Theorem \ref{T3B}.

\section{Application to multiple multicast network} \Label{S10}
Before considering the multiple multicast network,
we consider the multicast case, in which a single user broadcasts his/her message to several receivers.
This setting can be applied to the multiple multicast network that can be realized by a combination of point-to-point quantum key distribution.
Remember the encoding and the decoding depend only on 
the integers $m_0,m_1,m_2,m_3,m_4$ in Theorem \ref{T3}.
Assume that we have $b$ receivers (receiver$1$, $\ldots$, receiver $b$).
The integers $m_2$ and $m_3$ do not depend on receiver $i$.
Since $m_0, m_1$ and $m_4$ depend on receiver $i$,
they are written as $m_0(i), m_1(i)$ and $m_4(i)$.
Since each receiver can add an extra dimension,
we can consider that 
the dimension of observed information by each receiver is $\tilde{m}_4:=\max_i m_4(i)$.
Then, we choose 
$\tilde{m}_0:= \min_{i} m_0(i)$
and $\tilde{m}_1:= \max_{i} m_1(i)$.
We apply Theorem \ref{T3} to the case with $m_0= \tilde{m}_0$,
$m_1= \tilde{m}_1$, $m_2, m_3$, and $m_4= \tilde{m}_4$.
When the encoder of the obtained non-local code is used in the sender
and the decoder of the obtained non-local code is used in all the receivers,
the rate $\tilde{m}_0-\tilde{m}_1-{m}_2 $.
 
Now, we proceed to the multiple multicast case, which contains multiple unicast case.
We consider how to apply our result to
a multiple multicast network with $a$ senders and $ \sum_{i=1}^a b_i$ receivers, in which,
the senders and the receivers are labeled as $i$ and $(i,j)$ with $1\le i\le a$ and $1\le j \le b_i $, respectively.
Sender $i$ intends to securely send the message to Receiver $(i,j)$.
That is, Sender $i$ wants to keep secrecy for Receiver $(i',j)$ with $i'\neq i$.
In the one-time use of the network, 
Sender $i$ sends $m_{3,i}$ symbols $\vec{X}_i$ of $\FF_q$ via $m_{3,i}$ channels
and
Receiver $(i,j)$ receives $m_{4,i,j}$ symbols $\vec{Y}_{i,j}$of $\FF_q$ via $m_{4,i,j}$ channels.
Without loss of generality, 
adding extra dimensions,
we can assume that 
$m_{4,i,j}$ does not depend on $j$ and is simplified to $m_{4,i}$ due to the following reason.
That is, when $m_{4,i,j} < \max_{j'} m_{4,i,j'}$,
we can consider that Receiver $(i,j)$ receives symbol $0$ via $ \max_{j'} m_{4,i,j'}- m_{4,i,j} $ channels.
If the codes in the network are designed perfectly, we have no cross-line nor no information leakage to unintended receivers.
In this section, we consider the case with a small amount of cross line and information leakage to unintended receivers due to errors on the design of the network.

We assume that these senders and receivers are connected via network composed of linear operations.
Then, using matrices $K_{i,j;i'}$, we can describe their relations as
\begin{align}
\vec{Y}_{i,j}= \sum_{i'=1}^{a} K_{i,j;i'} \vec{X}_{i'}.
\end{align}
While the senders transmit their information repeatedly, 
we assume that the coefficient matrices $K_{i,j;i'}$ do not change.
We assume that receivers do not  collude to recover the message from senders.
Now, we apply the model \eqref{F1n} and \eqref{F2n} to
 the secure communication transmission from Sender $i$ to Receiver $(i,j)$.
When we consider information leakage to Receiver $(i'',j'')$ with $i''\neq i$,
we substitute $\vec{X}_i$, $(\vec{X}_{i'})_{i'\neq i}$, $\vec{Y}_{i,j}$ and $\vec{Y}_{i'',j''}$
into $X$, $Z$, $Y_B$, and $Y_E$, respectively. 
We assume that the rank of information crossed from other senders is $m_{1,i,j}$
and the rank of leaked information to Receiver $(i'',j'')$ is $m_{2,i;i'',j''}$.
We introduce the maximum ranks 
$m_{0,i}:= \max_{j}\rank K_{i,j;i}$, 
$m_{1,i}:= \max_{j}m_{1,i,j}$, and  $m_{2,i} := \max_{i'',j''} m_{2,i;i'',j''}$.
Sender $i$ and Receiver $(i,j)$ are assumed to know only the integers
$m_{0,i}$, $m_{1,i}$, $m_{2,i}$, $m_{3,i}$, $m_{4,i}$ and have no other knowledge for the network structure.
We choose our non-local code by applying Theorem \ref{T2B} to
the case with 
$m_0=m_{0,i}$,
$m_1=m_{1,i}$,
$m_2=m_{2,i}$,
$m_3=m_{3,i}$, and 
$m_4=m_{4,i}$.
Since the non-local code does not depend on the choice of $j$ and $i'',j''$, this non-local code works well in this situation.

\section{Conclusion}\Label{SCon}
We have discussed how sequential error injection affects the information leaked to Eve.
As the result, we have shown that 
there is no improvement when 
the network is composed of linear operations.
However, when the network contains non-linear operations,
we have found a counterexample that improves the information obtained by Eve.
Moreover, as Theorem \ref{T3},
we have shown the achievability of the asymptotic rate 
$m_{0}-m_{1}-m_{2}$ for a linear network under the secrecy and robustness conditions
when the transmission rate from Alice to Bob is $m_{0}$,
the rate of noise injected by Eve is $m_{1}$,
and the rate of information leakage to Eve is $m_{2}$.
The converse part of this rate is an interesting open problem.
In addition, 
as Theorem \ref{T3B},
we have discussed the secrecy and the asymptotic transmission rate 
when Eve has a possibility to inject noise into the network.

Further, we have applied our results to network quantum key distribution.
Then, we have clarified what type of network will enable us to realize secure long-distance communication based on short-distance quantum key distribution.
However, when we consider only the case given in Fig. \ref{F4},
we can employ a classical (non-quantum) secret sharing protocol \cite{Shamir} instead of network coding 
because all of the communications of this case are routing.
In particular, cheater-identifiable secret sharing against rushing cheaters \cite{PW91,RAX,XMT1,XMT2,HK}
enables us to share secure keys without using public channels or prior shared randomness.

In this way, this paper has discussed the application of secure network coding
to a network model whose communications on the edges 
are realized by quantum key distribution.
Replacing the role of quantum key distribution by physical layer security,
we can consider a secure network based on physical layer security.
In particular, 
we can use secure wireless communication \cite{H17,LPS,BC11,SC,Trappe,Zeng,WX,RealisticChannel,IFS1,IFS2}
as a typical form of physical layer security,
which provides us with a secure network based on secure wireless communication.
A crucial weak point of physical layer security is the possibility that the eavesdropper might break the assumption of the model.
Such an attack might be realized in the following cases.
(1) The eavesdropper concentrates his/her resources on one point.
(2) The eavesdropper luckily encounters a situation that the assumption is broken.
When we combine physical layer security and secure network coding in the above way,
to eavesdrop our information, 
the eavesdropper needs to break the model of physical layer security in multiple communication channels.
In case (1), to realize this condition, the eavesdropper has to 
distribute his/her resources, which increases the difficulty of eavesdropping.
For case (2), 
the eavesdropper must be lucky in multiple communication channels,
and this probability is very small.
In this way, this kind of combination is particularly useful.

\section*{Acknowledgments}
MH and NC are very grateful to Prof. Masaki Owari, Dr. Go Kato, Dr. Wangmei Guo,
and Mr. Seunghoan Song
for helpful discussions and comments.

\appendices

\section{Proof of Theorem \ref{T1}}\Label{ApB}
We define two random variables $Y_{E,p}^n:=K_E X^n$ and 
$Y_{E,a}^n:=K_E X^n+ H_E Z^n$.
Eve can simulate her outcome $Y_{E,a}^n$ under the strategy $\alpha$
from $Y_{E,p}^n$, $H_E$, and $Z^n$.
So, we have 
$I(M;Y_{E,p}^n) \ge I(M;Y_{E,a}^n,Z^n)=I(M;Y_{E,a}^n)$.
Conversely, since $Y_{E,p}^n$ is given as a function of $Y_{E,a}^n$, $Z^n$, and $H_E$,
we have the opposite inequality.

\section{Proof of Proposition \ref{T2}}\Label{Ap1}

To show Proposition \ref{T2}, 
we regard any element of the finite field $\FF_q$ as
an element of a $t$-dimensional algebraic extension $\FF_{q'}$ 
of the finite field $\FF_q$, where $q'=q^t$.
The matrices $K_B,H_B,K_E,H_E$ on $\FF_q$ can be regarded as matrices on $\FF_{q'}$.
By choosing $l_n:=t n$, the matrices $X^{l_n}$, $Y_B^{l_n}$, $Y_E^{l_n}$, $Z^{l_n}$ on $\FF_q$ 
are converted to matrices ${X'}^{n}$, ${Y_B'}^{n}$, ${Y_E'}^{n}$, ${Z'}^{n}$ on $\FF_{q'}$,
which also satisfy \eqref{F1n} and \eqref{F2n} by regarding 
the same matrices $K_B,H_B,K_E,H_E$ on $\FF_q$ as matrices on $\FF_{q'}$.
Then, the following proposition is known.

\begin{proposition}[\cite{JL,JLKHKM,Jaggi2008,Yao2014}]
\Label{T2B}
We assume the following two conditions for $m_2, m_1, m_0$ and a sequence of prime power $q_{n}'$.
The inequality $m_{2}+m_{1}< m_{0}$ holds.
The size $q_{n}'$ of the finite field increases such that 
$\frac{q_{n}'}{{n}^{m_0+1}} \to \infty$. 
Then, there exists a sequence of non-local codes $\Phi_{n}$ 
of block-length $n$ on finite field $\FF_{q_{n}'}$
whose message set is $\FF_{q_{n}'}^{k_{n}'}$ such that
\begin{align}
&\lim_{n \to \infty} \frac{k_{n}'}{n} = m_{0}-m_{1}\Label{eqH4-27C}
\\
&\lim_{n \to \infty} \max_{ \bm{K},\bm{H}} \max_{\alpha^n} 
P_e[\Phi_{n}, \bm{K},\bm{H},\alpha^n]=0 \Label{eqH4-27},
\end{align}
where the maximum is taken in the same way as with Proposition \ref{T2}.
\end{proposition}

The optimality of the rate $m_0-m_{1}$ was also shown 
under the condition \eqref{eqH4-27} in \cite[Sections VI \& VII]{Yao2014}. 
By choosing $t_{n}=\lceil \frac{ (m_0+1)\log n}{\log q}\rceil$
and $l_{n}:= t_{n} {n}$,
Proposition \ref{T2B} implies Proposition \ref{T2}.
Hence, we need to explain how to show Proposition \ref{T2B}. 

Combining the results in \cite{JL,JLKHKM,Jaggi2008,Yao2014},
we can construct a sequence of non-local codes to satisfy \eqref{eqH4-27C} and \eqref{eqH4-27}.
More precisely, the papers \cite[Section IX]{JLKHKM}\cite[Section VIII]{Jaggi2008}
constructed a sequence of non-local codes to satisfy \eqref{eqH4-27C} and \eqref{eqH4-27}
under  the condition $m_{2}+2 m_{1}< m_{0}$.
This is because the condition $m_{2}+2 m_{1}< m_{0}$ is stronger than 
the condition $m_{2}+m_{1}< m_{0}$, which is the assumption of Proposition \ref{T2B}.
To show \eqref{eqH4-27} under the weaker condition $m_{2}+m_{1}< m_{0}$, 
the papers 
by Jaggi, Langberg, Katti, Ho, Katabi, M\'{e}dard, and Effros
\cite[Section VII]{JLKHKM}\cite[Section VI]{Jaggi2008}\cite[Section IV-C]{Yao2014}
constructed a sequence of non-local codes to satisfy \eqref{eqH4-27C} and \eqref{eqH4-27},
when Alice can send Bob secret information whose size is asymptotically negligible 
in comparison with $n$, in the following way.

\begin{proposition}[\cite{JL,JLKHKM,Jaggi2008}]
\Label{T2C}
We assume the following three conditions.
The inequality $m_0>m_1$ holds.
Alice can send Bob secret information whose size is asymptotically negligible 
in comparison with $n$.
The size $q_{n}'$ of the finite field increases such that 
$\frac{q_{n}'}{{n}^{m_0+1}} \to \infty$.
Then, there exists a sequence of non-local codes $\Phi_{n}$ 
of block-length $n$ on finite field $\FF_{q_{n}'}$
whose message set is $\FF_{q_{n}'}^{k_{n}'}$ 
such that the relations \eqref{eqH4-27} and \eqref{eqH4-27C} hold. 
\end{proposition}

Then, under the weaker condition $m_{2}+ m_{1}< m_{0}$,
as the following proposition, the papers \cite[Section III]{JL}\cite[Section V]{Yao2014} provide a protocol for secure transmission of random variables with an asymptotically negligible length $k_{n}$ in comparison with $n$, which is the requirement in Proposition \ref{T2C}.

\begin{proposition}[\protect{\cite[Section III]{JL},\cite[Section V]{Yao2014}}]\Label{T2BL}
We assume the inequality $m_{2}+m_{1}< m_{0}$.
Then, there exists a sequence of non-local codes $\Phi_{n}$ of block-length $n$ 
whose message set is $\FF_{q}^{k_{n}}$ such that
\begin{align}
&\lim_{n \to \infty} k_{n}= \infty \Label{eqH4-27CY}
\\
&\lim_{n \to \infty} \max_{ \bm{K},\bm{H}} \max_{\alpha^n} 
P_e[\Phi_{n}, \bm{K},\bm{H},\alpha^n]=0 \Label{eqH4-27Y},\\
&\lim_{n \to \infty} \max_{ \bm{K},\bm{H}} 
I(M;Y_E^{n})[\Phi_{n},\bm{K},\bm{H},0]=0, \Label{H3-182Y}
\end{align}
where the maximum is taken in the same way as Proposition \ref{T2B}.
\end{proposition}

Therefore, attaching the protocol of Proposition \ref{T2BL} to the non-local codes given in 
Proposition \ref{T2C},
we obtain \eqref{eqH4-27} under the weaker condition $m_{2}+m_{1}< m_{0}$.
However, their proof of Proposition \ref{T2C}
is very hard to read because it omits the detail derivation. 
In the following, we give an alternative proof of Proposition \ref{T2C}.

Before our proof of Proposition \ref{T2C},
we prepare two lemmas.
The first lemma can be easily shown by the discussion of linear algebra.
In the following discussion, we simplify $q_{n}'$ to $q'$.

\begin{lemma}\Label{L4-27}
For integers $a_0 \le a_1+a_2,a_1',a_2'$, 
we fix an $a_1$-dimensional subspace $W_1 \subset \FF_{q'}^{a_1'}$ and 
an $a_2$-dimensional subspace $W_2 \subset \FF_{q'}^{a_2'}$.
We assume the following two conditions.
\begin{description}
\item[(E1)]
An $a_0 \times a_1'$ matrix $A_1$
and an $a_0 \times a_2'$ matrix $A_2$
satisfy 
\begin{align}
\Ker A_1|_{W_1} &=\{0\}, \\
\im A_1 \cap \im A_2 &= \{ 0 \},
\end{align}
where $\im (f)$ denotes the image of the function $f$.
\item[(E2)]
We consider a subspace $W_3 \subset W_1 \oplus W_2$.
For vectors
$x_1, \ldots, x_b \in W_1 $
and
$y_1, \ldots, y_b \in W_2$ with $b \ge a_1+a_2$,
$(x_1,y_1), \ldots, (x_b ,y_b) $ span $W_3$.
\end{description}
Then, we have the following statements.
\begin{description}
\item[(E3)]
There exists an $a_1' \times a_0 $ matrix $A_3$ such that
$ A_3 (A_1 x_i+ A_2 y_i)=x_i$ for $i=1, \ldots, b$, i.e., 
\begin{align}
&A_3 
\left[A_1~ A_2 \right]
\left[
\begin{array}{cccc}
x_1 & x_2 & \cdots & x_b \\
y_1 & y_2 & \cdots & y_b
\end{array}
\right]
\nonumber \\
=&
\left[
\begin{array}{cccc}
x_1 & x_2 & \cdots & x_b 
\end{array}
\right].
\end{align}
\item[(E4)]
The above matrix $A_3$ satisfies the relation
\begin{align}
 A_3 (A_1 x+ A_2 y)=x \Label{4-27F}
\end{align}
for any $(x,y) \in W_3$.
\end{description}
\end{lemma}

\begin{proof}
Due to condition (E1), 
we choose a map $A_4$ from $\im A_1$ to $W_1$
such that $A_4 A_1$ is the identify on $W_1$.
Since $W_1$ is included in $\FF_{q'}^{a_1'}$,
$A_4$ can be regarded as a map from $\im A_1$ to $\FF_{q'}^{a_1'}$.
Then, we choose a projection $A_5$ from $\FF_{q'}^{a_0}$
to $\im A_1$ such that $A_5 x=0$ for $x \in \im A_2$.
Therefore, $A_3:= A_5 A_4$ satisfies the condition of (E3).
Further, (E2) guarantees (E4).
\end{proof}


\begin{lemma}[\protect{\cite[Section VII]{JLKHKM}\cite[Claim 5]{Jaggi2008}}]\Label{LJa}
We independently choose $m$ random variables 
$V_1, \ldots, V_{m}$ 
subject to the uniform distribution on $\FF_{q'} $.
We define the $n \times m$ matrix $U_1$
as $U_{1;i,j}:= (V_j)^i$ with $i=1, \ldots, n$ and $j= 1,\ldots, m$.
Then,  
\begin{align}
\Pr \{ x U_1=x' U_1
\} \le \Big(\frac{n}{q'}\Big)^m
\end{align}
for any $x\neq x'\in \FF_{q'}^n$.
\end{lemma}

\begin{proofof}{Proposition \ref{T2C}}\par
\noindent{\bf Step (1): Non-local code construction}\par\noindent
First, we provide our non-local code when we use the channel $n$ times based on the finite field $\FF_{q'}$.
Our message is given as an $(m_{0}-m_{1}) \times n$ matrix $M$, which satisfies condition
\eqref{eqH4-27C} asymptotically. 
Since the rank of $H_B$ is $m_1$,
there exist an 
$m_4\times m_1$ matrix $\hat{H}_B$ and 
an $m_1\times n$ matrix $\hat{Z}^{n}$ such that
\begin{align}
Y_B^{n}
=K_B U_0  X^{n}+H_B Z^{n}
=K_B U_0  X^{n}+\hat{H}_B \hat{Z}^{n}.\Label{TT}
\end{align}
Then, we address $\hat{H}_B $ and $\hat{Z}^{n}$
instead of ${H}_B $ and ${Z}^{n}$.

We fix an integer $m:=m_0+1$.
We independently choose $m$ random variables 
$V_1, \ldots, V_{m}$ subject to the uniform distribution on $\FF_{q'} $.
Also, we randomly choose 
the $m_3 \times m_3$ matrix $U_0$ among all $m_3 \times m_3$ invertible matrices. 

Then, we define the $n \times m$ matrix $U_1$
as $U_{1;i,j}:= (V_j)^i$ with $i=1, \ldots, n$ and $j= 1,\ldots, m$.
We also define the $(m_{0}-m_{1})\times m$ matrix $U_2:= M U_1$.
Moreover, we define the $m_3 \times n $ matrix $X^{n}:=
\left[
\begin{array}{c}
M \\
0
\end{array}
\right]$, where $0$ is the $m_1 \times n$ zero matrix.
As secret information with a negligible rate,
Alice sends Bob the information $V_1, \ldots, V_{m},U_2$.
Then, Alice inputs the $m_3 \times n $ matrix $U_0 X^{n}$ 
as the input of $n$ times use of the channel.

Then, Bob receives the $m_4 \times n$ matrix $Y_B^{n}$ given in \eqref{TT}
as well as the secret information $V_1, \ldots, V_{m},U_2$.
Since the ranks of $K_B U_0 X^{n}$ and $\hat{H}_B$ are $m_{0}-m_{1}$ and $m_1$ at most, respectively,
the rank of the matrix $Y_B^{n}$ is $m_0$ at most. 
We denote the rank by $\bar{m}_0$.
We choose $\bar{m}_0$ linearly independent row vectors from the row vectors of $Y_B^{n}$.
We denote the $\bar{m}_0 \times n$ matrix composed of the $\bar{m}_0 $ 
independent row vectors by $\bar{Y}_B^{n} $.
Similarly, we denote the matrices composed of these $\bar{m}_0$ row vectors of 
the matrices $K_B$ and $\hat{H}_B$ by
$\bar{K}_B$ and $\bar{H}_B$, respectively.
Then, using the standard Gaussian elimination,
Bob finds a matrix $ U_3$ to satisfy the equation
\begin{align}
U_3 \bar{Y}_B^{n} U_1= U_2,
\Label{4-27Y}
\end{align}
which is equivalent to
\begin{align}
U_3 
 \left[
\bar{K}_BU_0 P_{m_0-m_1} ~ \bar{H}_B 
\right]
 \left[
\begin{array}{c}
M \\
\hat{Z}^{n}
\end{array}
\right] U_1=M U_1,\Label{EEQ1}
\end{align}
where 
$P_{m_0-m_1}$ is the imbedding
$\left[
\begin{array}{c}
I \\
0\end{array}
\right]
 $
from $\FF_{q'}^{m_0-m_1} $ to 
$\FF_{q'}^{m_3} $.
Notice that Bob can calculate 
$U_1$ from the secret information $V_1, \ldots, V_{m}$.
Finally, Bob recovers the information $\hat{M}:=U_3 \bar{Y}_B^{n}$.
To check the condition \eqref{4-27Y}, Bob needs only 
$\bar{Y}_B^{n}$,
$U_2$, and $ U_1$, which can be computed from 
$V_1, \ldots, V_{m}$.

\noindent{\bf Step (2): Analysis of performance}\par\noindent
There are two conditions if the above protocol is to work well.
\begin{description}
\item[(F1)]
The relations 
$\im (\bar{K}_B U_0 P_{m_0-m_1})\cap \im \bar{H}_B= \{0\}$
and $\Ker \bar{K}_B U_0 P_{m_0-m_1}|_{\im M} =\{0\}$ hold.
\item[(F2)]
The relation 
$
\rank  
 \left[
\begin{array}{c}
M \\
\hat{Z}^{n}
\end{array}
\right]
 U_1
=\rank 
\left[
\begin{array}{c}
M \\
\hat{Z}^{n}
\end{array}
\right]
$ holds,
where $\rank $ denotes the rank of the matrix.


\end{description}

Assume that conditions (F1) and (F2) hold.
We apply Lemma \ref{L4-27} to the case when
$a_0= \bar{m}_0,
a_1=\rank M,
a_2= \rank \hat{Z}^{n},
W_1=\im M,
W_2= \im \hat{Z}^{n},
A_1=\bar{K}_B U_0 P_{m_0-m_1},
A_2=\bar{H}_B,
A_3= U_3$.
Then, conditions (F1) and (F2) guarantee
conditions (E1) and (E2), respectively.
Then, due to condition (E3),
there exists a matrix $ U_3$ that satisfies equation \eqref{4-27Y}, i.e., \eqref{EEQ1}.
Condition (E4) guarantees that
$U_3 \bar{Y}_B^{n}= M$, i.e., 
Bob can decode the message $M$.

Now, we evaluate the probability that condition (F2) holds.
Condition (F2) holds if and only if 
$ z^T  \left[
\begin{array}{c}
M \\
\hat{Z}^{n}
\end{array}
\right]
 U_1 \neq 0$ for any $z \in \FF_{q'}^{m_0}$
 satisfying the condition $
 z^T  \left[
\begin{array}{c}
M \\
\hat{Z}^{n}
\end{array}
\right]
\neq 0 $.
Applying Lemma \ref{LJa} to all of $z (\neq 0)\in \FF_{q'}^{m_0}$,
we find that condition (F2) holds at least with probability 
$1- {q'}^{m_0} (\frac{n}{q'})^m
=1- \frac{{n}^m}{{q'}^{m-m_0}}=1- \frac{{n}^{m_0+1}}{{q'}} \to 1$.

Finally, we evaluate the probability that condition (F1) holds.
As shown later, the following conditions (F1') and (F1'') imply condition (F1).
\begin{description}
\item[(F1')]
The relation $\im (K_B U_0 P_{m_0-m_1})\cap \im \hat{H}_B= \{0\}$ holds.
\item[(F1'')]
The relation $\Ker {K}_B U_0 P_{m_0-m_1}|_{\im M} =\{0\}$ holds.
\end{description}
Hence, we show that conditions (F1') and (F1'') hold with a probability close to $1$.
Condition $\im (K_B U_0 P_{m_0-m_1})\cap \im \hat{H}_B= \{0\}$ holds if and only if
$ \im U_0 P_{m_0-m_1}\cap K_B^{-1} (\im \hat{H}_B)= \{0\}$.
For a fixed $K_B,\hat{H}_B$, 
since $\dim K_B^{-1} (\im \hat{H}_B) \le m_3-m_0+m_1$,
the probability of 
condition $\im (K_B U_0 P_{m_0-m_1})\cap \im \hat{H}_B= \{0\}$ is at least
\begin{align}
& (1-{q'}^{m_3-m_0+m_1-m_3})
(1-{q'}^{m_3-m_0+m_1-m_3+1})\cdots \nonumber \\
&\cdot (1-{q'}^{m_3-m_0+m_1-m_3+ m_0-m_1-1}) \nonumber \\
=&(1-{q'}^{m_1-m_0})
(1-{q'}^{m_1-m_0+1})\cdots
(1-{q'}^{-1})
\nonumber \\
=& 1-O(1/q').
\end{align}

The relation $\Ker {K}_B U_0 P_{m_0-m_1}|_{\im M} =\{0\}$ holds if and only if
no basis of $\im U_0 P_{m_0-m_1} M$ belongs to the space $\Ker {K}_B $.
Since $\Ker {K}_B $ is an $m_3-m_0$-dimensional subspace of an $m_3$-dimensional space,
the probability of condition 
$\Ker {K}_B U_0 P_{m_0-m_1}|_{\im M} =\{0\}$ is
\begin{align}
&(1-{q'}^{-m_0})
(1-{q'}^{-m_0+1})\cdots
(1-{q'}^{-m_0+m_7-1})\nonumber \\
=&1-O({q'}^{-m_0+m_7-1}),
\end{align}
where $m_7:= \dim \im P_{m_0-m_1} M(= \dim \im U_0  P_{m_0-m_1} M) 
=\dim \im M= \rank M \le m_0-m_1$.
Therefore, since $q'$ is sufficiently large,
we obtain the desired statement.

Finally, we show that condition (F1') implies condition (F1).
Since the relation
$\im (\bar{K}_B U_0 P_{m_0-m_1})\cap \im \bar{H}_B= \{0\}$
can be shown easily from (F1'), we show only the relation
$\Ker \bar{K}_B U_0 P_{m_0-m_1}|_{\im M} =\{0\}$ from (F1') and (F1'').
The choice of $\bar{m}_0$ guarantees that 
there exists an invertible map $U_4$ from $\im [K_B U~ \hat{H}_B]$ to $\FF_{q'}^{\bar{m}_0}$ such that
$U_4 [K_B U~ \hat{H}_B]=[\bar{K}_B U~ \bar{H}_B]$. 
Thus, 
\begin{align}
&\Ker \bar{K}_B U_0 P_{m_0-m_1}|_{\im M} =
\Ker U_4^{-1} \bar{K}_B U_0 P_{m_0-m_1}|_{\im M} 
\nonumber \\
=&
\Ker {K}_B U_0 P_{m_0-m_1}|_{\im M} =\{0\}.
\end{align}
\end{proofof}

\begin{remark}
Our proof is different from the proof presented in \cite[Section VII]{JLKHKM}\cite[Section VI]{Jaggi2008}\cite[Section IV-C]{Yao2014}.
They suggested that $ (m_0-m_1)m_0+1$ be chosen as $m$ 
because they employ the concept of list decoding.
However, our discussion allows us to choose a much smaller value $m_0+1$
as $m$.
This fact shows that our evaluation is better than their evaluation in this sense.
Note that our evaluation does not use list decoding.
\end{remark}

Proposition \ref{T2B} requires a finite field $\FF_{q'}$ with an infinitely large $q'$.
The paper \cite[Appendix D]{Hayashi-Tsurumaru} discussed the construction of $\FF_{2^t}$ whose
multiplication and inverse multiplication have calculation complexity $O(t \log t)$\footnote{The multiplication of elements $v$ and $z$ of $\FF_{2^t}$ is essentially given in (124) of \cite{Hayashi-Tsurumaru}
by using Fourier transform via a calculation on circulant matrices.
For the inverse multiplication of an element $v$ of $\FF_{2^t}$, 
we calculate $ F^{-1}[-Fv. *Fz]$ instead of 
$ F^{-1}[Fv. *Fz]$ in (124), where $F$ is discrete Fourier transform.}.

\end{document}